\documentclass[conference,letterpaper]{IEEEtran}
\usepackage[utf8]{inputenc} 
\usepackage[T1]{fontenc}
\usepackage{url}
\usepackage{ifthen}
\usepackage{cite}
\usepackage[cmex10]{amsmath}
\usepackage{amssymb}
\usepackage{xcolor}
\usepackage{graphicx}
\usepackage[font=small,labelfont={bf,it}, textfont=it]{caption}
\usepackage{subcaption}
\usepackage{amsthm}
\theoremstyle{plain}
\newtheorem{thm}[]{Theorem}

\theoremstyle{plain}
\newtheorem{cor}{Corollary}

\theoremstyle{definition}
\newtheorem{defn}{Definition}

\theoremstyle{remark}

\DeclareMathOperator*{\E}{\mathbb{E}}
\let\P\relax
\DeclareMathOperator*{\P}{\mathbb{P}}
\DeclareMathOperator*{\R}{\mathbb{R}}
\DeclareMathOperator*{\N}{\mathbb{N}}

\begin{document}
\title{One-Shot Wyner--Ziv Compression of a Uniform Source} 

\author{
  \IEEEauthorblockN{Oğuzhan Kubilay Ülger and Elza Erkip}
  \IEEEauthorblockA{Dept. of Electrical and Computer Engineering\\
                    New York University, New York, USA\\
                    Email: \{kubi, elza\}@nyu.edu}
}

\maketitle

\begin{abstract}
   In this paper, we consider the one-shot version of the classical Wyner-Ziv problem where a source is compressed in a lossy fashion when only the decoder has access to a correlated side information. Following the entropy-constrained quantization framework, we assume a scalar quantizer followed by variable length entropy coding. We consider compression of a uniform source, motivated by its role in the compression of processes with low-dimensional features embedded within a high-dimensional ambient space. We find upper and lower bounds to the entropy-distortion functions of the uniform source for quantized and noisy side information, and illustrate tightness of the bounds at high compression rates. 
\end{abstract}

\section{Introduction}

In their celebrated work, Slepian and Wolf \cite{SW} showed the surprising result that lossless distributed compression of correlated sources can be as efficient as their joint compression. This result is asymptotic and relies on a random binning argument. This random binning technique has also been used to establish many achievability theorems in asymptotic distributed source coding problems \cite{WZ,elgamalBook}. 

One of the most well-studied distributed lossy compression problems is the one by Wyner and Ziv\cite{WZ}. In the so-called Wyner--Ziv setting, a source is compressed in lossy fashion while a correlated side information (SI) is only available at the decoder. The focus in \cite{WZ} is the classical asymptotic blocklength setting.

Motivated by practical compression techniques that operate in the finite blocklength regime, in this paper, we are interested in the one-shot version of the Wyner--Ziv problem where the encoder first quantizes a single realization of the source into a countable set and then uses variable length lossless coding to turn it into a bitstream. The decoder wishes to reconstruct the source with the help of a correlated SI. We will call this setting, shown in Fig.~\ref{fig:1b}, the one-shot decoder-only SI case. The goal is to minimize the expected length of codewords under some average distortion constraint. This approach differs from nonasymptotic fixed length compression considered in \cite{kostina1, poisson}, where the rate solely depends on the number of quantization bins.

It is well known that when prefix-free lossless codes are used in the point-to-point (no SI) one-shot compression problem (Fig.~\ref{fig:1a}), the minimum expected codeword length is lower bounded by the entropy of the quantized version of the source. Furthermore, with optimal entropy coding such as Huffman coding \cite{cover}, the expected codeword length becomes at most 1-bit away from this lower bound. This 1-bit gap vanishes asymptotically as many quantized samples are entropy coded jointly. Hence looking at the entropy of quantizer output is a good metric for judging the actual compression rate. This observation led to the popular class of lossy compressors known as entropy-constrained scalar \cite{sullivanLap} and vector \cite{philChou} quantizers, and the associated entropy--distortion bounds \cite{gyorgyUnif,wagnerSawbridge,wagnerRamp}, studied in detail in the point-to-point setting of Fig.~\ref{fig:1a}. However, there is limited work in the literature that incorporates SI~\cite{WZ_Servetto,WZ_Zhenyu}. 

\begin{figure}[t]
\centering
\begin{subfigure}[b]{0.40\textwidth}
   \includegraphics[width=1\linewidth]{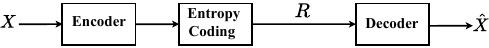}
   \caption{}
   \label{fig:1a} 
\end{subfigure}
\begin{subfigure}[b]{0.40\textwidth}
   \includegraphics[width=1\linewidth]{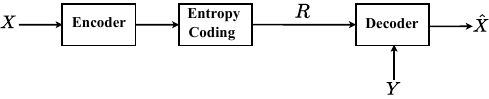}
   \caption{}
   \label{fig:1b}
\end{subfigure}
\begin{subfigure}[b]{0.40\textwidth}
   \includegraphics[width=1\linewidth]{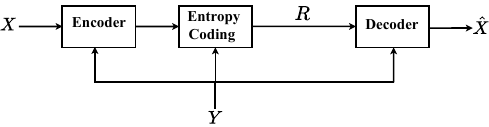}
   \caption{}
   \label{fig:1c}
\end{subfigure}
\caption{One-shot lossy compression systems considered in this paper: (a)  point-to-point, (b) decoder-only side information,  (c) both encoder and decoder side information.}
\end{figure}

Recently Wagner and Ball\'e\cite{wagnerSawbridge} and Bhadane et. al.~\cite{wagnerRamp} investigated compression of processes that are characterized by a low-dimensional manifold structure within a high-dimensional ambient space such as images \cite{balleImage}. As argued in~\cite{wagnerSawbridge,wagnerRamp}, comparing entropy--distortion functions of these processes with operational rate-distortion bounds of learned compressors allows uncovering capabilities of practical neural image compressors. It is shown in \cite{wagnerSawbridge,wagnerRamp} that compression of the aforementioned processes is equivalent to compression of a uniform source under different distortion metrics. Motivated by this observation and the emergence of neural compressors in decoder-only SI settings \cite{ezgiICML,ezgiISIT,ezgiJsait}, we study lossy one-shot compression of a uniform source with SI.  We consider two different SI models. In the first one, the SI represents a coarsely quantized version of the source to be refined. In the second one, we consider a noisy version of the source as the SI. For each SI, we characterize and provide lower and upper bounds for entropy--distortion trade-offs. We believe that this study can be a first step in establishing entropy-distortion bounds for distributed compression of processes like those in~\cite{wagnerSawbridge,wagnerRamp}, with potential implications on the design of distributed neural compressors for natural sources such as images.

The rest of the paper is organized as follows. In Section~\ref{sec:Sys} we formally define our problem and setting. In Section~\ref{sec:Main} we present our results on the entropy--distortion trade-off for a uniform source with SI. Section~\ref{sec:disc} finishes the paper with an interpretation of results and discussion.

\section{System Model} \label{sec:Sys}

Before considering the SI setting, we first present definitions for encoder, decoder and entropy--distortion function as in \cite{wagnerSawbridge, wagnerRamp} which considered point-to-point compression as shown in Fig.~\ref{fig:1a}.

\begin{defn}
\label{def:one-shot}
    Let $X$ be a random variable with distribution $P_X$ defined on the alphabet $\mathcal{X}$ and $\mathcal{\hat{X}}$ be the reconstruction alphabet. A  point-to-point one-shot compression of $X$ under distortion metric $d:\mathcal{X} \times \mathcal{\hat{X}} \rightarrow \R_{\geq 0}$ consists of an encoder $f : \mathcal{X} \rightarrow \N$ and a decoder $g: \N \rightarrow \mathcal{\hat{X}}$. The entropy of this encoder-decoder pair is given by
    \begin{align*}
        H(f) &= -\sum_{i}\P[f(X)=i]\log(\P[f(X)=i])
    \end{align*}
    and the distortion $D(f,g) \!= \!\E[d(X,\hat{X})]$ where $\hat{X} = g(f(X))$.
\end{defn} 

\begin{defn}
\label{def:one-shot-ed}
    Given a source $X \sim P_X$ and a distortion metric $d(\cdot,\cdot)$, the point-to-point entropy--distortion function is given by
    \begin{align*}
        H(\Delta) = \inf_{f} H(f) 
    \end{align*}
    where infimum is taken over all encoders $f$ such that there exists a decoder $g$ with $D(f,g) \leq \Delta$ for some $\Delta \in \R_{\geq 0}$, with $f$ and $g$ as in Definition~\ref{def:one-shot}.
\end{defn}

A natural extension of Definitions~\ref{def:one-shot} and~\ref{def:one-shot-ed} to the decoder-only SI case in Fig.~\ref{fig:1b} considers entropy coding after quantization, similar to the point-to-point setting. This is reflected in the definitions below.

\begin{defn}
\label{def:dSI}
    Let $X$ be a random variable representing the source, $Y$ be a correlated random variable representing the SI available only at the decoder with joint distribution $P_{XY}$  on the product alphabet $\mathcal{X} \times \mathcal{Y}$, and $\mathcal{\hat{X}}$ be the reconstruction alphabet. A one-shot compression of $X$ with decoder-only SI $Y$ under distortion metric $d:\mathcal{X} \times \mathcal{\hat{X}} \rightarrow \R_{\geq 0}$ consists of an encoder $f : \mathcal{X} \longrightarrow \N$ and a decoder $g : \N \times \mathcal{Y} \longrightarrow \mathcal{\hat{X}}$. The entropy of this encoder-decoder pair is given by
    \begin{align*}
        H(f) = -\sum_{i}\P[f(X)=i]\log(\P[f(X)=i])
    \end{align*}
    and the distortion $D_{SI}(f,g) = \E[d(X,\hat{X})]$ where $\hat{X} =\ g(f(X),Y)$.
\end{defn}
\begin{defn}
\label{def:dSI-ed}
    Given a source $X$, SI $Y$ and distortion metric $d(\cdot,\cdot)$ the entropy--distortion function with decoder-only SI is given by
    \begin{align*}
        H(\Delta) = \inf_{f} H(f) 
    \end{align*}
    where the infimum is taken over all encoders $f$ such that there exists a decoder $g$ with $D_{SI}(f,g) \leq \Delta$ for some $\Delta \in \R_{\geq 0}$, with $f$ and $g$ as in Definition~\ref{def:dSI}.
\end{defn}

As an alternative to entropy coding, one could use a lossless Slepian--Wolf compressor following the encoder in Fig.~\ref{fig:1b}, resulting in rate $H(f|Y)$ (see~\cite{SW}) in Definition~\ref{def:dSI} rather than $H(f)$.  While practical codes for the Slepian--Wolf problem are available~\cite{WZ_Pradhan,SW-Prac4,SW-Prac3,SW-Prac1}, their performances are characterized only for specific source distributions, and their use in practice is limited compared with entropy coding \cite{philChou,sullivanLap}. Therefore, for entropy--distortion function with decoder-only SI, we will follow Definition~\ref{def:dSI-ed} rather than the one in \cite{ezgiJsait} that assumes the availability of Slepian--Wolf coding. 

If the SI is available both at the encoder and the decoder as in Fig.~\ref{fig:1c}, then both quantization and entropy coding can use the particular realization of $Y$. This leads to the following definitions.

\begin{defn} \label{def:c1}
    Let $X$ be a random variable representing the source, $Y$ be a correlated random variable representing the SI available both at the encoder and decoder with joint distribution $P_{XY}$  on the product alphabet $\mathcal{X} \times \mathcal{Y}$, and $\mathcal{\hat{X}}$ be the reconstruction alphabet. A conditional one-shot compression of $X$ with  SI $Y$ under distortion metric $d:\mathcal{X} \times \mathcal{\hat{X}} \rightarrow \R_{\geq 0}$ consists of an encoder $f : \mathcal{X} \times \mathcal{Y} \longrightarrow \N$ and a decoder $g : \N \times \mathcal{Y} \longrightarrow \mathcal{\hat{X}}$. The conditional entropy of this encoder-decoder pair is given by
    \begin{align*}
        &H(f|Y)\!= \E\biggl[ \!-\!\!\sum_{i}\P[f(X)\!=\!i|Y\!=\!y]\log(\P[f(X)\!=\!i|Y\!=\!y]) \!\biggl]
    \end{align*}
    and the distortion $D_{C}(f,g)\!=\! \E[d(X,\hat{X})]$ where $\hat{X} = g(f(X,Y),Y)$.
\end{defn}
\begin{defn} \label{def:c2}
    Given a source $X$, SI $Y$ and distortion metric $d(\cdot,\cdot)$, the conditional entropy--distortion function is given by
    \begin{align*}
        H_{C}(\Delta) = \inf_{f} H(f|Y)
    \end{align*}
    where the infimum is taken over all encoders $f$ such that there exists a decoder $g$ with $D_{C}(f,g) \leq \Delta$ for some $\Delta \in \R_{\geq 0}$, with $f$ and $g$ as in Definition~\ref{def:c1}.
\end{defn}
Note that $H_C(\Delta)\leq H_{SI}(\Delta)$ since the encoder in Definition~\ref{def:c1} can always ignore the particular realization of $Y$ and only use its statistics as in Definition~\ref{def:dSI}. 

In the rest of this paper, we consider a source uniformly distributed on the interval $[0,1]$, i.e. $X \sim \text{Unif}([0,1])$. In \cite{wagnerSawbridge}, the authors show that compressing the random process called "sawbridge", which models processes with low-dimensional features embedded within a high-dimensional ambient space such as images, under squared error is equivalent to compressing the uniform source under $L_1$ distance, $d(x,\hat{x}) = |x-\hat{x}|$. Therefore, we will assume the $L_1$ distance as our distortion metric. However, we note that our results can be easily generalized to $d(x,\hat{x}) = |x-\hat{x}|^\rho$ for all $\rho \geq 1$.

We consider two different SI scenarios. In the first setting, we assume that the SI is a coarsely quantized version of the source. Specifically, the SI is $Y_q$ where
\begin{equation} \begin{split} \label{eq_qside}
    \P[Y_q = k] = \frac{1}{K}, \;\; k \in \{1,2,\ldots,K\}, \\
    P_{X|Y_q}(x | Y_q = k) \sim \text{Unif}\left(\left[\frac{k-1}{K},\frac{k}{K}\right]\right),
\end{split}
\end{equation}
for some integer $K \geq 1$. In the second setting, the SI is a noisy observation of the source
\begin{align*}
    Y_n = X + Z \; (\text{mod } 1)
\end{align*}
where $Z$ is independent of $X$ and distributed uniformly on $[-\alpha/2,\alpha/2]$ for some $\alpha \in [0,1/2]$. In both settings, given the SI, the source is distributed uniformly in a subset of $[0,1]$. Note the quantized version of the SI is a deterministic function of the source and hence can be obtained directly from $X$. In that case, without loss of generality, we assume $Y_q$ is also available at the encoder, leading to $H_{SI}(\Delta)=H_C(\Delta)$. Also, even though we consider a uniformly quantized SI, $Y_q$, our results in Theorem~\ref{thm:qach} hold for general non-uniformly quantized SI with a slight modification to the optimization problem.

\section{Main Results} \label{sec:Main}
We first present the point-to-point entropy--distortion function of a uniform source under $L_1$ distance as shown in \cite{gyorgyUnif} and \cite{wagnerSawbridge}.

\begin{thm}[\cite{gyorgyUnif,wagnerSawbridge}] \label{ThmPtP}
Consider the source $X \sim \text{Unif}([0,1])$ and $L_1$ distortion metric $d(x,\hat{x}) = |x-\hat{x}|$. Then the entropy distortion function of $X \sim \text{Unif}([0,1])$ is given by
    \begin{align*}
         H^U(\Delta) = -\biggl\lfloor \frac{1}{p} \biggl\rfloor p \log p - q \log q,  \quad &0 < \Delta \leq 1/4
    \end{align*}
    where $q = (1 - \big\lfloor \frac{1}{p} \big\rfloor p)$ and $p\in(0,1)$ is the unique solution to 
    \begin{align*}
        \biggl\lfloor \frac{1}{p}\biggl\rfloor p^2 + q^2 = 4\Delta.
    \end{align*}
\end{thm}
It is shown in \cite{gyorgyUnif} that the encoder and decoder that achieve the entropy--distortion function in Theorem~\ref{ThmPtP} is obtained by dividing the unit interval into $N$ intervals with $N-1$ intervals of size $p$ and one interval of size $q$. The reconstructions are the mid-points of their corresponding intervals.

\subsection{Quantized Side Information}
In this section, we consider the quantized SI $Y_q$ which is a deterministic function of $X$. Given a realization of $Y_q$ as in (\ref{eq_qside}), both the encoder and the decoder know that $X$ is distributed uniformly on an interval of length $1/K$. Thus the encoder only needs to quantize $X$ over this interval. Gyorgy and Linder~\cite{gyorgyUnif} show that Theorem~\ref{ThmPtP} holds for any uniform distribution over an arbitrary interval after rescaling of quantization intervals and the distortion constraint. We use this in the next theorem to provide the conditional entropy--distortion function for the quantized SI.

\begin{thm}[Quantized Side Information] \label{thm:qach}
    Consider a $\text{Unif}([0,1])$ source and quantized SI model described in (\ref{eq_qside}) for some integer $K\geq 1$. The distortion metric is $L_1$, that is $d(x,\hat{x}) = |x-\hat{x}|$. Then the entropy--distortion function with decoder-only SI and the conditional entropy--distortion functions are equal, and are given by
    \begin{equation} \label{eq:opt}\begin{split}
            H_{SI}^q(\Delta) = H_{C}^q(\Delta) = &\min_{\{\Delta_k\}_{k=1}^{K}} \; \frac{1}{K}\sum_{k=1}^{K}H^U(K\Delta_k) \\
            \text{s.t. }\; &\frac{1}{K}\sum_{k=1}^{K}\Delta_k \leq \Delta, \\ &\Delta_k \geq 0 \text{ for all } k\in \{1,\ldots,K\},
        \end{split}
    \end{equation}
    where $H^U(\cdot)$ is the entropy--distortion function of the uniform source given in Theorem~\ref{ThmPtP}.
\end{thm}
\begin{proof}
By \cite{gyorgyUnif}, one shot optimal point-to-point encoder and decoder for $\Tilde{X}\sim\text{Unif}([a,b])$ are given as
\begin{align*}
    \Tilde{f}(x) &= f\left(\frac{x-a}{b-a}\right) \text{ and } \Tilde{g}(i) = (b-a)g(i)+a
\end{align*}
where $f$ and $g$ are optimal one-shot point-to-point encoder-decoder pair for $X\sim \text{Unif}([0,1])$ with distortion $\Delta$. Furthermore 
\begin{equation} \label{eq_scu} \begin{split}
    H(\Tilde{f}) &= H(f) = H^U(\Delta) \\
    D(\Tilde{f},\Tilde{g}) &=(b-a)\Delta.
\end{split}
\end{equation}
For any encoder and decoder of the conditional entropy--distortion with quantized SI problem $f(X,Y_q),g(f(X,Y_q),Y_q)$, let $f_k(X) = f(X,Y_q=k)$ and $g_k(f_k(X)) = g(f(X,Y_q=k),Y_q=k)$. Then the conditional entropy and distortion of $f,g$ can be written as
\begin{align*}
    H(f|Y) &= \frac{1}{K} \sum_{k=1}^{K}H^U(f_k) \\
    D(f,g) &= \frac{1}{K}\sum_{k=1}^{K}\E[d(X,\hat{X})|Y_q=k] = \frac{1}{K} \sum_{k=1}^{K}D(f_k,g_k)
\end{align*}
Note that the pair $(f_k,g_k)$ is tasked with quantizing a uniform random variable ($X$ given $Y_q=k$) over the interval $[(k-1)/K,k/K]$. Then under the constraint $D(f_k,g_k) = \Delta_k$, by (\ref{eq_scu}) the smallest entropy is $H(f_k)$ is $H^U(K\Delta_k)$. So given any collection of $\{\Delta_k\}_{k=1}^{K}$, the optimal entropy and distortion are given by
\begin{align*}
    H(f|Y) &= \frac{1}{K}\sum_{k=1}^{K}H^U(K\Delta_k) \\
    D(f,g) &= \frac{1}{K} \sum_{k=1}^{K}\Delta_k.
\end{align*}
Optimizing $H(f|Y)$ over all $\Delta_k$ satisfying $D(f,g) \leq \Delta$ completes the proof.
\end{proof}

The optimization problem in Theorem~\ref{thm:qach} is a non-convex one. In fact $H^U(\cdot)$ is strictly concave on the intervals $(\frac{1}{4(N+1)},\frac{1}{4N})$ for all integers $N\geq 1$ thus the minimizing the sum  is not trivial. The following corollary gives easy-to-compute upper and lower bounds for $H_{C}^{q}(\Delta)$.

\begin{cor} \label{cor:q}
    The conditional/decoder-only SI entropy--distortion function for the uniform source with quantized SI in Theorem~\ref{thm:qach} can be upper and lower bounded as
    \begin{align*}
       \Breve{H}^U(K\Delta) \leq   H_C^q(\Delta) \leq  H^U(K\Delta)
    \end{align*}
    where $H^U(\cdot)$ is the entropy--distortion function of the uniform source given in Theorem~\ref{ThmPtP} and $\Breve{H}^U(\cdot)$ is its the convex envelope.
\end{cor}
\begin{proof}
    The upper bound can be obtained directly by setting $\Delta_k = \Delta$ for all $k \in \{1,\ldots,K\}$ for the optimization problem given in Theorem~\ref{thm:qach}. For the lower bound assume that $\Delta_k^*$ is a minimizer of the optimization problem in Theorem~\ref{thm:qach} ($k \in \{1,\ldots,K\}$). Then
    \begin{align*}
        H_C^q(\Delta) = \frac{1}{K} \sum_{k=1}^{K}H^U(K\Delta_k^*) \geq \frac{1}{K} \sum_{k=1}^{K}\Breve{H}^U(K\Delta_k^*).
    \end{align*}
    by definition of the convex envelope. Also
    \begin{equation*} \label{eq:opt}\begin{split}
            \frac{1}{K} \sum_{k=1}^{K}\Breve{H}^U(K\Delta_k^*) \geq \min_{\{\Delta_k\}_{k=1}^{K}} \; &\frac{1}{K}\sum_{k=1}^{K}\Breve{H}^U(K\Delta_k) \\
            \text{s.t } &\frac{1}{K}\sum_{k=1}^{K}\Delta_k \leq \Delta \\
            &\Delta_k \geq 0 \text{ for all } k\in \{1,\ldots,K\}
        \end{split}
    \end{equation*}
    This optimization problem is a symmetric and convex minimization problem and thus is $\Delta_k = \Delta$ for all $k$, is a minimizer \cite[Exercise 4.4]{boyd2004convex}. Thus
    \begin{align*}
        H_C^q(\Delta) \geq \Breve{H}^U(K\Delta)
    \end{align*}
    which completes the proof.
\end{proof}

The upper and lower bounds in Corollary~\ref{cor:q} meet at rates equal to $\log N$ for some integer $N\geq 1$ and do not meet for any other rate \cite{gyorgyUnif}. Furthermore, it can be shown that the lower bound is tight at $K+1$ points between the rates $\log N$ and $\log (N+1)$. This can be achieved by choosing $M$ out of $K$ distortions $\Delta_k$ such that $H^U(K\Delta_k) = \log N$ and rest of the $K-M$ distortions such that $H^U(K\Delta_k) = \log (N+1)$. As a result, the entropy will be $\frac{M}{K} \log N + \frac{K-M}{K} \log (N+1)$ and we will achieve the convex envelope.
\subsection{Noisy Side Information}
When the SI is of the form $Y_n = X + Z \text{ (mod 1)}$, given a realization of the SI $Y_n = y$, $X$ is distributed uniformly over the set
\begin{equation} \label{eq:By}
    B_y = \begin{cases}
        [y-\alpha/2,y+\alpha/2], \;\;\;\;\;\;\;\;\;\;\;\;\;\;\;\;\; \alpha/2 \leq y \leq 1-\alpha/2,& \\
        [0,1] \setminus  (y+\alpha/2-1,y-\alpha/2), \;\;\;\; 1-\alpha/2 < y \leq 1,& \\
        [0,1] \setminus (y+\alpha/2, y-\alpha/2 + 1), \;\;\;\;\;\;\;\;\;\; 0 \leq y < \alpha/2.&
    \end{cases}
\end{equation}
If the SI was available at both sides, it would be sufficient for the encoder to only consider the set $B_y$ for quantization for each $y$, and the problem would be similar to that with quantized SI. On the other hand, dividing $[0,1]$ into intervals with separate encoder indices for each interval is wasteful because the decoder can readily use $B_y$ to distinguish between intervals more than $\alpha$ apart. 

For the decoder-only SI case, similar to "binning" in the asymptotic Wyner--Ziv problem, our achievability result for noisy SI divides $[0,1]$ into many intervals and then groups intervals that are sufficiently apart into a single encoder output. The intervals are such that for each $y$, only one of the intervals in the group has a nonempty intersection with $B_y$. This grouping ensures that using the SI, the decoder can identify a single interval out of multiple disjoint intervals that correspond to the same encoder output.

\begin{thm}[Achievability for Noisy Side Information] \label{thm:nach}
    Consider the source $X\sim \text{Unif}[0,1]$ and SI $Y_n = X + Z \text{ (mod 1)}$ where $Z$ is independent of $X$ and distributed uniformly on $[-\alpha/2,\alpha/2]$ for some $\alpha \in [0,1/2]$. Then the entropy--distortion function with decoder-only SI $H_{SI}^{n}(\Delta)$ is upper bounded by
    \begin{align*}
        H_{SI}^{n}(\Delta) &\leq -\biggl\lfloor \frac{1}{p} \biggl\rfloor p \log p - q \log q
    \end{align*}
    where $q = (1 - \big\lfloor \frac{1}{p} \big\rfloor p)$ and $p\in(0,1-\alpha)$ is the solution to
    \begin{align*}
        12\Delta &= \biggl\lfloor \frac{1}{p} \biggl\rfloor \left(3p \left(\frac{p}{L} \wedge \alpha \right) - \frac{L}{\alpha}\left(\frac{p}{L} \wedge \alpha \right)^3\right) \\
        &+ \left(3q \left(\frac{q}{L} \wedge \alpha \right) - \frac{L}{\alpha}\left(\frac{q}{L} \wedge \alpha \right)^3\right)
    \end{align*}
    with $L = \lfloor (1-p)/\alpha\rfloor$ and $(a \wedge b) = \min(a,b)$.
\end{thm}
\begin{proof}
For any $p \in (0,1-\alpha)$, let $N = \big\lfloor \frac{1}{p} \big\rfloor +1$. Then consider an encoder that has range of $\{1,\ldots,N\}$ such that
\begin{align*}
    \P[f(X) = i] &= p \;\text{for}\; i = 1,\ldots,N-1 \\
    P[f(X) = N] &= 1 - \lfloor 1/p \rfloor p \triangleq q.
\end{align*}
The entropy of such an encoder is
\begin{align*}
    H(f) = -\lfloor1/p\rfloor p \log p - q \log q.
\end{align*}
Now for $L = \lfloor (1-p)/\alpha\rfloor$ which is the maximum number of intervals with combined size $p$ that can be placed at least $\alpha$ apart, consider the following quantization of the interval $[0,1/L]$:
\begin{align*}
    f(x) = \begin{cases}
        i, \quad  \frac{(i-1)p}{L} \leq x < \frac{ip}{L}, i = 1,\ldots,N-1 \\
        N, \quad  \frac{(N-1)p}{L} \leq x \leq \frac{1}{L}.
    \end{cases}
\end{align*}
For the rest of the interval we let $f$ be periodic, i.e. $f(x+\ell/L) = f(x)$ for $\ell = 1,\ldots,L-1$ and $x \in [0,1/L]$. Note that by this construction there are $L$ disjoint intervals on $[0,1]$ that are mapped to the same encoder index. Furthermore, these intervals are at least $\alpha$ apart.

Letting $A_i = \{x: f(x) = i\}$, by construction $A_i$ is given by the union of $L$ disjoint intervals of size $p/L$ for $i \in \{1,\ldots,N-1\}$ and size $q/L$ for $i=N$. For any $i$, let $A_i = \bigcup_{\ell=1}^{L}C_{i,\ell}$ where $C_{i,\ell}$ are the aforementioned intervals. Specifically
\begin{align*}
    C_{i,\ell} = \begin{cases}
        \left[\frac{(i-1)p}{L} + \frac{\ell-1}{L}, \frac{ip}{L} + \frac{\ell-1}{L} \right), \;\;\; &i \in \{1,\ldots,N-1\} \\
        \left[\frac{(N-1)p}{L} + \frac{\ell-1}{L}, \frac{\ell}{L} \right], \;\;\; &i = N.
    \end{cases}
\end{align*}
When index $i$ is received the decoder maps the encoder output to the mid-point of the intersection $A_i \cap B_y$. For any $y \in [0,1]$ and $i \in \{1,\ldots,N\}$, $C_{i,\ell} \cap B_y \neq \emptyset$ for at most one $\ell \in \{1,\ldots,L\}$. Hence $A_i \cap B_y$ is an interval for all $y \in [0,1]$ and $i \in \{1,\ldots,N\}$.

Now we compute the expected distortion of this encoder-decoder pair. We first note for an arbitrary interval $S = [a,b]$ with $0 \leq a < b \leq 1$ with $\hat{X} = (a+b)/2$ we get
\begin{align*}
    \E[|X-\hat{X}||X\in S] = \frac{1}{b-a}\int_a^{b}|x-(a+b)/2|dx = \frac{b-a}{4}
\end{align*}
We can write the expected distortion since $f(X)=i$ is equivalent to $X \in A_i$, we can write
\begin{align*}
    \E[|X-\hat{X}|] &= \sum_{i=1}^{N}\P(f(X)=i)\E[|X-\hat{X}| | X \in A_i].
\end{align*}
Now for each $A_i$ since $C_{i,\ell}$ are disjoint and equal sized,
\begin{align*}
    \E[|X-\hat{X}| | X \in A_i] = \frac{1}{L}\sum_{\ell = 1}^{L}\E[|X-\hat{X}| | X \in C_{i,\ell}]
\end{align*}
and each term of the summation:
\begin{align*}
    \E[&|X-\hat{X}||X \in C_{i,\ell}] = \int_{0}^{1}\E[|X-\hat{X}| | X \in C_{i,\ell},Y=y]\\
    &\qquad \qquad \qquad \qquad \qquad \qquad \qquad \quad \cdot f_{Y|X\in C_{i,\ell}}(y|C_{i,\ell})dy \\
    &= \int_{0}^{1}\frac{|C_{i,\ell}\cap B_y|}{4}f_{Y|X\in C_{i,\ell}}(y|C_{i,\ell})dy
\end{align*}
where $f_{Y|X\in C_{i,\ell}}(y|C_{i,\ell})$ is the conditional probability density function of $Y$ given $X \in C_{i,\ell}$. Now by symmetry of $Z$, shifting $C_{i,\ell}$ doesn't change the above integral. Hence without loss of generality assume $C_{i,\ell} = [\alpha/2,p/L+\alpha/2]$ for $i \in \{1,\ldots,N-1\}$ and the conditional probability density can be written as
\begin{align*}
    &f_{Y|X\in C_{i,\ell}}(y|C_{i,\ell}) = \\
    &\frac{L}{\alpha p} \begin{cases} y, \quad  &0 \leq y \leq (p/L \wedge \alpha) \\
        (p/L \wedge \alpha), \quad &(p/L \wedge \alpha) < y \leq (p/L \vee \alpha) \\
        \alpha + p/L - y, \quad &(p/L \vee \alpha) < y \leq p/L + \alpha \\
        0, \quad &\text{otherwise}.
    \end{cases}
\end{align*}
where $(a \vee b) = \max(a,b)$. Then combining with $B_y$ in (\ref{eq:By}) we can write
\begin{align*}
    \E[|X-\hat{X}||X \in C_{i,\ell}] = (p/L \wedge \alpha) - \frac{L}{3\alpha p}(p/L \wedge \alpha)^3.
\end{align*}
We can find a similar result for $i=N$ by replacing $p$ with $q$. Putting it all together we get the desired expected distortion $\Delta$.
\end{proof}
To obtain a converse on the entropy--distortion function with decoder-only SI, we assume the noisy SI is available at the encoder as well, and  use $H_{SI}^{n}(\Delta) \geq H_{C}^{n}(\Delta)$. 

\begin{thm}[Converse for Noisy Side Information] \label{thm:nconv}
    Consider the source $X\sim \text{Unif}[0,1]$ and SI $Y_n = X + Z \text{ (mod 1)}$ where $Z$ is independent of $X$ and distributed uniformly on $[-\alpha/2,\alpha/2]$ for some $\alpha \in [0,1/2]$. Then the entropy--distortion function with decoder-only SI $H_{SI}^{n}(\Delta)$ is lower bounded as
    \begin{equation} \begin{split}
        H_{SI}^{n}(\Delta) \geq \inf_{\delta(y)} \;&\int_{0}^{1} H^{U}\left(\frac{\delta(y)}{\alpha}\right)dy \label{eq:conv2} \\
        \text{s.t } \; &\int_{0}^{1}\delta(y)dy \leq \Delta \\
        &\delta(y) \geq 0 \text{ for all } y \in [0,1].
    \end{split}
    \end{equation}
    where $H^{U}(\cdot)$ is the entropy--distortion function of the uniform source given in Theorem~\ref{ThmPtP}.
\end{thm}

\begin{proof}
See Appendix~\ref{app:nconv}
\end{proof}

Similar to the lower bound in Corollary~\ref{cor:q}, we can further relax the converse in Theorem~\ref{thm:nconv} and obtain a simpler bound.

\begin{cor} \label{cor:n}
    The entropy--distortion function with decoder-only SI for the uniform source with noisy SI, $H_{SI}(\Delta)$ can be lower bounded as
    \begin{align*}
         H_{SI}^{n}(\Delta) \geq \Breve{H}^{U}(\Delta/\alpha)
    \end{align*}
    where $\Breve{H}^{U}(\cdot)$ is the convex envelope of $H^{U}(\cdot)$ given in Theorem~\ref{ThmPtP}.
\end{cor}
\begin{proof}
The proof is similar to that of Corollary~\ref{cor:q} and is provided in  Appendix~\ref{app:corn}. 
\end{proof}

\begin{figure}[t]
\centering
\begin{subfigure}[b]{0.5\textwidth}
   \includegraphics[width=1\linewidth]{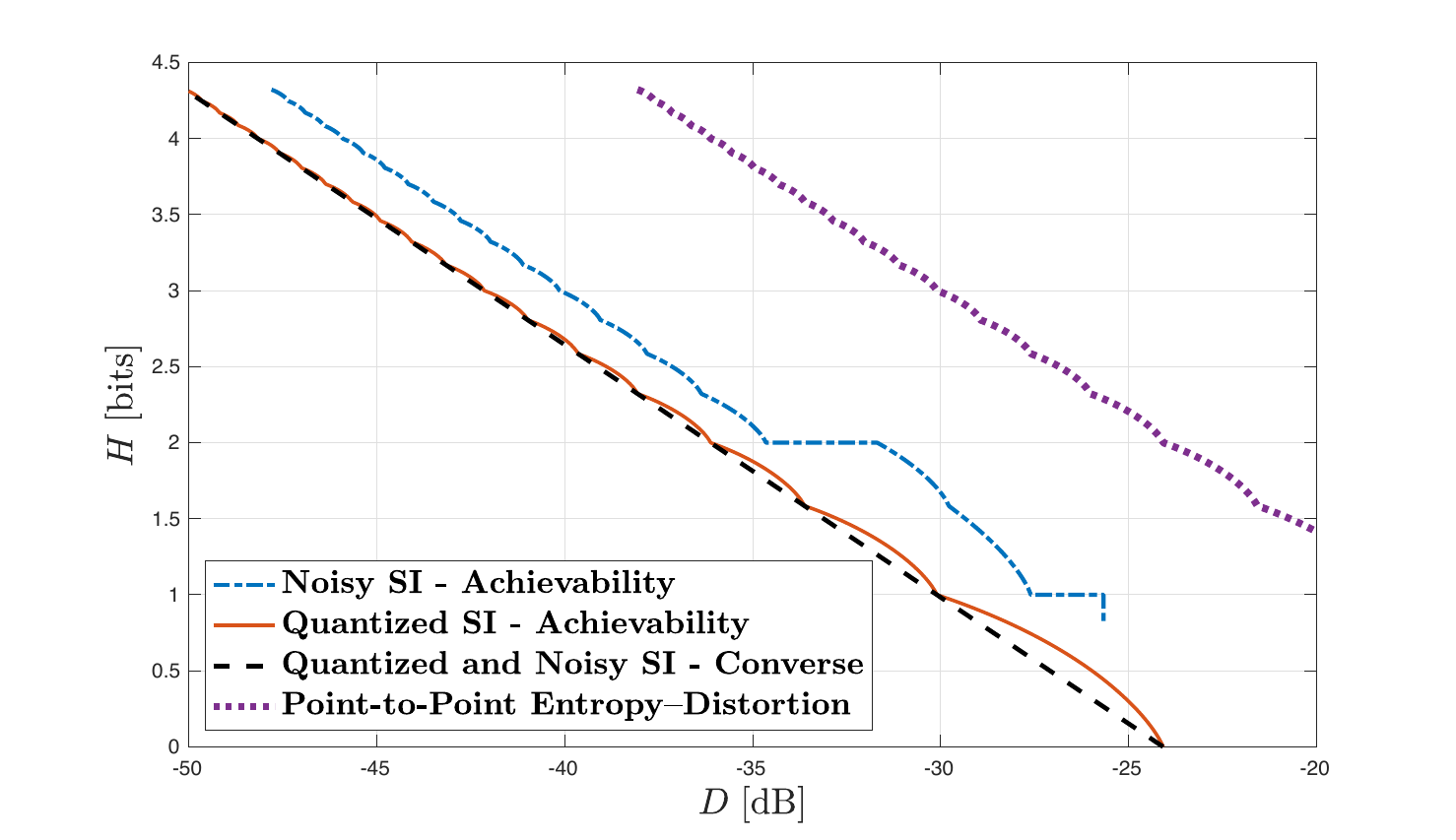}
   \caption{}
   \label{fig:ex1} 
\end{subfigure}
\begin{subfigure}[b]{0.5\textwidth}
   \includegraphics[width=1\linewidth]{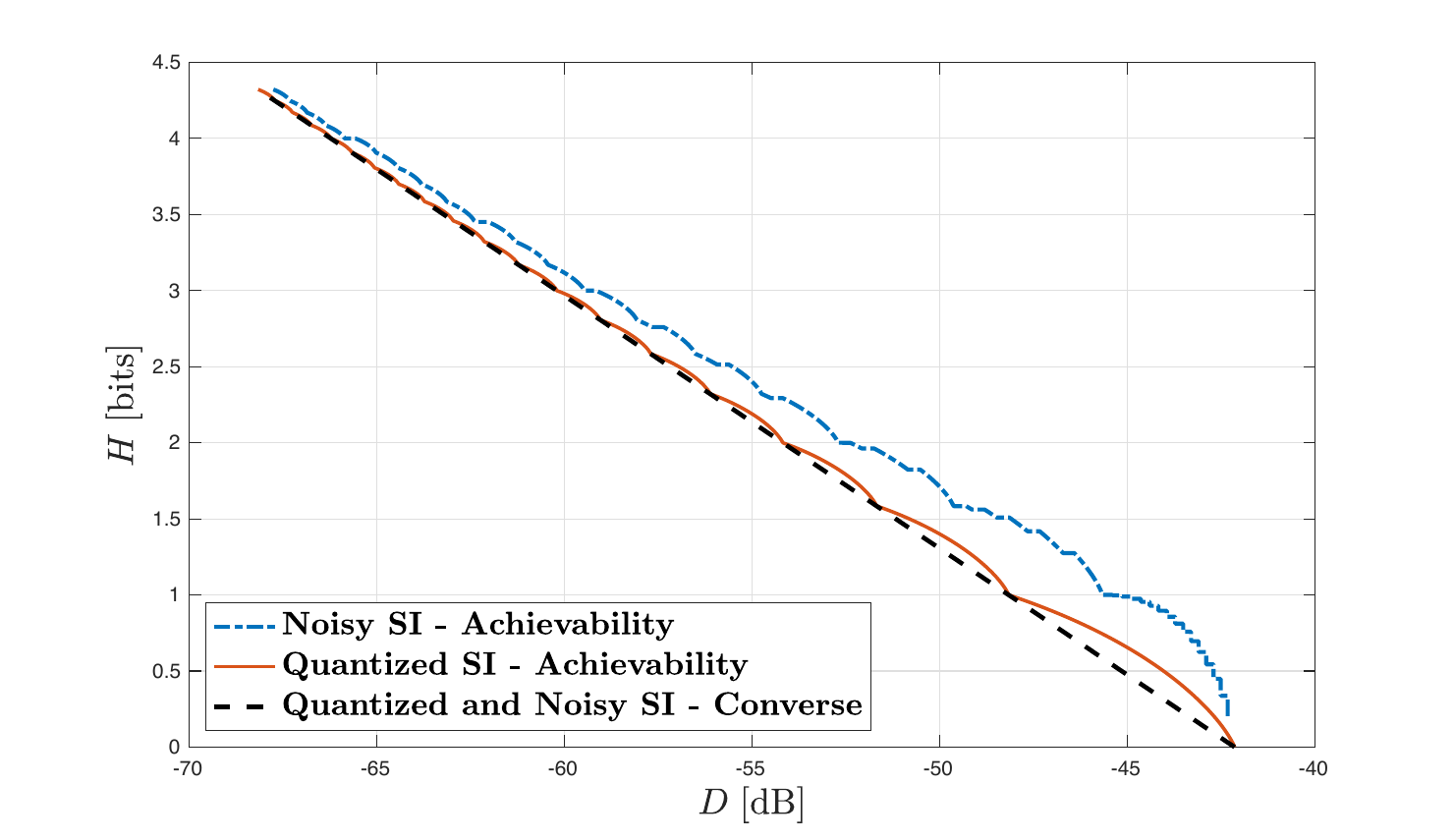}
   \caption{}
   \label{fig:ex2}
\end{subfigure}
\caption{Upper and lower bounds of entropy--distortion trade-off for quantized and noisy SI case for (a) $K=1/\alpha=4$  and (b) $K=1/\alpha=32$.}
\label{fig:exb}
\end{figure}

\section{Illustrations and Discussion} \label{sec:disc}
We first note that the lower bounds in Corollary~\ref{cor:q} and~\ref{cor:n} are in the same form. In fact when $\alpha = 1/K$, these bounds are equal. The case with $\alpha = 1/K$ is of special interest because for both noisy and quantized SI models, this corresponds to $X$ given $Y$ being uniformly distributed on a set with size $1/K$. In the case of quantized SI, these sets are uniformly placed and disjoint on the unit interval, and in the noisy SI case the set moves continuously as the value of the SI $Y$ changes. 

Fig.~\ref{fig:exb} shows comparison of upper bounds on $H_{SI}^{n}(\Delta)$ (Theorem~\ref{thm:nach}) and $H_C^{q}(\Delta)$ (Corollary~\ref{cor:q}), and the lower bound for both (Corollary~\ref{cor:q} and~\ref{cor:n}) for the cases where $K=1/\alpha=4$  (Fig.~\ref{fig:ex1}) and $K=1/\alpha=32$ (Fig.~\ref{fig:ex2}). Note that $H^U(\Delta)$ is strictly concave between rates $(\log N,\log (N+1))$ for all integers $N\geq 1$ and it only meets its convex envelope $\Breve{H}^U(\Delta)$ at rates $\log N$, resulting in the difference between the lower and upper bounds of $H_C^{q}(\Delta)$ given in Corollary~\ref{cor:q}. The gap between these bounds gets smaller at higher rates. Also as $K$ increases, i.e. the correlation between $X$ and $Y$ increases, the difference between upper and lower bounds decreases at high rates. Similarly, the achievability result in Theorem~\ref{thm:nach} for $H_{SI}^{n}(\Delta)$, is closer to the lower bound at higher rates and has a smoother behavior. Due the nature of the construction we have used in the achievability of $H_{SI}^{n}(\Delta)$, a change in the number of intervals per encoder index $(L)$ causes abrupt changes in the resultant distortion. In the case of $K=4$ shown in Fig.~\ref{fig:ex1}, at high rates we start with $L=\lfloor(1-p)K\rfloor=3$, and the boundaries for $L=2$ and $L=1$ can be seen as horizontal lines. As expected there is a substantial difference between entropy--distortion functions of point-to-point case and the ones with SI.

To conclude, in this work we have investigated a one-shot Wyner--Ziv problem for uniform sources with two SI models. We have presented upper and lower bounds for the entropy--distortion functions and showed that they get tighter at higher rates. We believe that potentially, these results can be used to investigate entropy--distortion trade-off for more complex processes as in \cite{wagnerRamp,wagnerSawbridge} in the setting with SI.

\bibliographystyle{IEEEtran}
\bibliography{ref}

\begin{thebibliography}{10}
\providecommand{\url}[1]{#1}
\csname url@samestyle\endcsname
\providecommand{\newblock}{\relax}
\providecommand{\bibinfo}[2]{#2}
\providecommand{\BIBentrySTDinterwordspacing}{\spaceskip=0pt\relax}
\providecommand{\BIBentryALTinterwordstretchfactor}{4}
\providecommand{\BIBentryALTinterwordspacing}{\spaceskip=\fontdimen2\font plus
\BIBentryALTinterwordstretchfactor\fontdimen3\font minus \fontdimen4\font\relax}
\providecommand{\BIBforeignlanguage}[2]{{%
\expandafter\ifx\csname l@#1\endcsname\relax
\typeout{** WARNING: IEEEtran.bst: No hyphenation pattern has been}%
\typeout{** loaded for the language `#1'. Using the pattern for}%
\typeout{** the default language instead.}%
\else
\language=\csname l@#1\endcsname
\fi
#2}}
\providecommand{\BIBdecl}{\relax}
\BIBdecl

\bibitem{SW}
D.~Slepian and J.~Wolf, ``Noiseless coding of correlated information sources,'' \emph{IEEE Transactions on Information Theory}, vol.~19, no.~4, pp. 471--480, 1973.

\bibitem{WZ}
A.~Wyner and J.~Ziv, ``The rate-distortion function for source coding with side information at the decoder,'' \emph{IEEE Transactions on Information Theory}, vol.~22, no.~1, pp. 1--10, 1976.

\bibitem{elgamalBook}
A.~El~Gamal and Y.-H. Kim, \emph{Network Information Theory}.\hskip 1em plus 0.5em minus 0.4em\relax Cambridge university press, 2011.

\bibitem{kostina1}
V.~Kostina and S.~Verd{\'u}, ``Fixed-length lossy compression in the finite blocklength regime,'' \emph{IEEE Transactions on Information Theory}, vol.~58, no.~6, pp. 3309--3338, 2012.

\bibitem{poisson}
C.~T. Li and V.~Anantharam, ``A unified framework for one-shot achievability via the {Poisson} matching lemma,'' \emph{IEEE Transactions on Information Theory}, vol.~67, no.~5, pp. 2624--2651, 2021.

\bibitem{cover}
T.~M. Cover, \emph{Elements of Information Theory}.\hskip 1em plus 0.5em minus 0.4em\relax John Wiley \& Sons, 1999.

\bibitem{sullivanLap}
G.~J. Sullivan, ``Efficient scalar quantization of exponential and laplacian random variables,'' \emph{IEEE Transactions on Information Theory}, vol.~42, no.~5, pp. 1365--1374, 1996.

\bibitem{philChou}
P.~A. Chou, T.~Lookabaugh, and R.~M. Gray, ``Entropy-constrained vector quantization,'' \emph{IEEE Transactions on Acoustics, Speech, and Signal Processing}, vol.~37, no.~1, pp. 31--42, 1989.

\bibitem{gyorgyUnif}
A.~Gyorgy and T.~Linder, ``Optimal entropy-constrained scalar quantization of a uniform source,'' \emph{IEEE Transactions on Information Theory}, vol.~46, no.~7, pp. 2704--2711, 2000.

\bibitem{wagnerSawbridge}
A.~B. Wagner and J.~Ball{\'e}, ``Neural networks optimally compress the sawbridge,'' in \emph{2021 Data Compression Conference (DCC)}.\hskip 1em plus 0.5em minus 0.4em\relax IEEE, 2021, pp. 143--152.

\bibitem{wagnerRamp}
S.~Bhadane, A.~B. Wagner, and J.~Ball{\'e}, ``Do neural networks compress manifolds optimally?'' in \emph{2022 IEEE Information Theory Workshop (ITW)}.\hskip 1em plus 0.5em minus 0.4em\relax IEEE, 2022, pp. 582--587.

\bibitem{WZ_Servetto}
S.~D. Servetto, ``Lattice quantization with side information: Codes, asymptotics, and applications in sensor networks,'' \emph{IEEE Transactions on Information Theory}, vol.~53, no.~2, pp. 714--731, 2007.

\bibitem{WZ_Zhenyu}
Z.~Tu, T.~J. Li, and R.~S. Blum, ``On scalar quantizer design with decoder side information,'' in \emph{2006 40th Annual Conference on Information Sciences and Systems}, 2006, pp. 224--229.

\bibitem{balleImage}
O.~J. H{\'e}naff, J.~Ball{\'e}, N.~C. Rabinowitz, and E.~P. Simoncelli, ``The local low-dimensionality of natural images,'' \emph{arXiv preprint arXiv:1412.6626}, 2014.

\bibitem{ezgiICML}
\BIBentryALTinterwordspacing
E.~Ozyilkan, J.~Ballé, and E.~Erkip, ``Neural distributed compressor does binning,'' in \emph{2023 ICML Workshop Neural Compression: From Information Theory to Applications}, 2023. [Online]. Available: \url{https://openreview.net/forum?id=3Dq4FZJSga}
\BIBentrySTDinterwordspacing

\bibitem{ezgiISIT}
------, ``Learned {W}yner–{Z}iv compressors recover binning,'' in \emph{2023 IEEE International Symposium on Information Theory (ISIT)}, 2023, pp. 701--706.

\bibitem{ezgiJsait}
------, ``Neural distributed compressor discovers binning,'' \emph{IEEE Journal on Selected Areas in Information Theory}, pp. 1--1, 2024.

\bibitem{WZ_Pradhan}
S.~Pradhan and K.~Ramchandran, ``Distributed source coding using syndromes ({DISCUS}): {D}esign and construction,'' \emph{IEEE Transactions on Information Theory}, vol.~49, no.~3, pp. 626--643, 2003.

\bibitem{SW-Prac4}
Z.~Liu, S.~Cheng, A.~Liveris, and Z.~Xiong, ``Slepian-{W}olf coded nested lattice quantization for {W}yner-{Z}iv coding: High-rate performance analysis and code design,'' \emph{IEEE Transactions on Information Theory}, vol.~52, no.~10, pp. 4358--4379, 2006.

\bibitem{SW-Prac3}
A.~Liveris, Z.~Xiong, and C.~Georghiades, ``Compression of binary sources with side information at the decoder using {LDPC} codes,'' \emph{IEEE Communications Letters}, vol.~6, no.~10, pp. 440--442, 2002.

\bibitem{SW-Prac1}
Z.~Tu, J.~Li, and R.~Blum, ``Compression of a binary source with side information using parallelly concatenated convolutional codes,'' in \emph{IEEE Global Telecommunications Conference, 2004. GLOBECOM '04.}, vol.~1, 2004, pp. 46--50.

\bibitem{boyd2004convex}
S.~P. Boyd and L.~Vandenberghe, \emph{Convex {O}ptimization}.\hskip 1em plus 0.5em minus 0.4em\relax Cambridge university press, 2004.

\end{thebibliography}

\appendices

\section{Proof of Theorem~\ref{thm:nconv}} \label{app:nconv}
It is straightforward to see for all $\Delta > 0$, $H_{SI}^{n}(\Delta) \geq H_C^{n}(\Delta)$. Now we will bound $H_C^{n}(\Delta)$ for the noisy SI setting. Similar to the proof of Theorem~\ref{thm:qach}, for any encoder and decoder of the conditional entropy--distortion with quantized SI problem $f(X,Y_n),g(f(X,Y_n),Y_n)$, let $f_y(X) = f(X,Y_n=y)$ and $g_y(f_y(X)) = g(f(X,Y_n=y),Y_n=y)$. Then the conditional entropy and distortion of $f,g$ can be written as
\begin{align*}
    H(f|Y) &= \int_{0}^{1}H^U(f_y)dy \\
    D(f,g) &= \int_{0}^{1}\E[d(X,\hat{X})|Y_n=y]dy = \int_{0}^{1}D(f_y,g_y)dy
\end{align*}
Note that the pair $(f_y,g_y)$ is tasked with quantizing a uniform random variable ($X$ given $Y_n=y$) over the interval $B_y$ in (\ref{eq:By}) with $|B_y|=\alpha$. For $\alpha/2 \leq y \leq 1-\alpha/2$, $B_y$ is a single interval but for other values of $y$, $B_y$ takes the form of a union of two disjoint intervals. Consequently, we first show that the optimal achievable entropy--distortion function for a uniform source over two disjoint intervals is lower bounded by entropy--distortion function of a uniform source over a single interval with the same length.

Suppose that $X\sim\text{Unif}([0,a]\cup[b,c])$ where $c\geq b \geq a$. Let $f$ be any encoder that quantizes $[0,a]\cup[b,c]$. Then for some encoding index $i$, denote $S_i = \{x : f(x)=i\}$ with $\P[X \in S_i] = p_i$, and the decoder output $g(i)=t_i$. For $L_1$ distortion metric, the median of $S_i$ minimizes the expected distortion $\E[|X-t_i||X\in S_i]$, hence $t_i$ is set to be the median of $S_i$, i.e $|\{x: x \leq t_i, x \in S_i\}| = |\{x: x \geq t_i, x \in S_i\}|$. Note that $t_i$ does not need to be unique or belong to the set $[0,a]\cup[b,c]$, and if it is not unique then all medians give the same expected distortion. Contribution of this encoding index to $H(f)$ is $-p_i\log(p_i)$ and its contribution to the the distortion is
\begin{align*}
    \E&[|X-g(f(X))||X \in S_i] = \E[|X-t_i||X \in S_i] \\
    &=\frac{1}{p_i}\int_{S_i \cap [0,a]}|x-t_i|dx + \frac{1}{p_i}\int_{S_i \cap [b,c]}|x-t_i|dx
\end{align*}

Now let, $\Tilde{X}\sim\text{Unif}([0,a+(c-b)])$ which has the same length as the support of $X$ but in a single interval. Further let $\Tilde{S}_i= (S_i\cap[0,a])\cup(\{x : (x+ (b-a)) \in (S_i \cap [b,c] \})$ so $\Tilde{S}_i \subseteq[0,a+(c-b)]$ and $\P[\Tilde{X} \in \Tilde{S_i}] = p_i$. Define an encoder $\Tilde{f}$, such that $\Tilde{f}(\Tilde{x}) = i$ iff $x \in \Tilde{S}_i$. Repeating this for all $i$, we get an encoder $\Tilde{f}$ that compresses $\Tilde{X}$ over the interval $[0,a+(b-c)]$. Since $\P[X \in S_i] = \P[\Tilde{X} \in \Tilde{S}_i]$ for all $i$, $H(f)=H(\Tilde{f})$. 

Denote the median of $\Tilde{S_i}$ by $\Tilde{t}_i$. Then if $t_i \in [0,a] \Rightarrow \Tilde{t}_i = t_i$ and the expected distortion induced by $\Tilde{S}_i$
\begin{align*}
    \E&[|\Tilde{X}-\Tilde{t}_i||\Tilde{X} \in \Tilde{S}_i] = \frac{1}{p_i}\int_{ \cap [0,a+(c-b)]}|x-\Tilde{t}_i|dx\\
    &=\frac{1}{p_i}\int_{\Tilde{S}_i \cap [0,a]}|x-\Tilde{t}_i|dx + \frac{1}{p_i}\int_{\Tilde{S}_i \cap [a,a+(c-b)]}(x-\Tilde{t}_i)dx \\
    &\leq \frac{1}{p_i}\int_{\Tilde{S}_i \cap [0,a]}|x-t_i|dx + \frac{1}{p_i}\int_{S_i \cap [b,c]}(x-t_i)dx \\
    &= \E[|X-t_i||X \in S_i]
\end{align*}
where the inequality follows from the fact that $\Tilde{t}_i = t_i$ and $S_i \cap [b,c]$ is obtained by shifting $\Tilde{S}_i \cap [a,a+(c-b)]$ to right by $b-a$. Similarly, in the case $t_i \in [b,c]$ we have $a \leq \Tilde{t}_i = t_i -(b-a)$. Hence

\begin{align*}
    \frac{1}{p_i}\int_{\Tilde{S}_i \cap [0,a]}|x-\Tilde{t}_i|dx &= \frac{1}{p_i}\int_{\Tilde{S}_i \cap [0,a]}(\Tilde{t}_i-x)dx \\
    \leq \frac{1}{p_i}\int_{S_i \cap [0,a]}(t_i-x)dx &= \frac{1}{p_i}\int_{S_i \cap [0,a]}|t_i-x|dx
\end{align*}
since $a \leq \Tilde{t}_i \leq t_i$  and
\begin{align*}
    \frac{1}{p_i}\int_{S_i \cap [b,c]}|x-t_i|dx =  \frac{1}{p_i}\int_{\Tilde{S}_i \cap [a,a+(c-b)]}|x-\Tilde{t}_i|dx
\end{align*}
by change of variable $x$ to $x + (b-a)$ of left hand side. We do not need to consider the case $t_i \notin [0,a]\cup [b,c]$ since we can just change it to $t_i = a$ or $t_i=b$ stemming from the fact that any choice of median results in the same distortion for $L_1$ distortion metric. Hence we have shown that for any encoder-decoder pair that works for the union of two disjoint intervals, we can construct encoder-decoder pair that compresses a single interval of the same size as the union of the two intervals with the same entropy and lower distortion.

So for all $B_y$, under the constraint $D(f_y,g_y) = \delta(y) \geq 0$, by (\ref{eq_scu}) the entropy $H(f_y)$ is lower bounded by $H^U(\delta(y)/\alpha)$. So given any integrable function $\delta(y)$, optimal achievable entropy and distortion are lower bounded by
\begin{align*}
    H(f|Y) &\geq  \int_{0}^{1}H^U(\delta(y)/\alpha)f_Y(y)dy = \int_{0}^{1}H^U(\delta(y)/\alpha)dy \\
    D(f,g) &=  \int_{0}^{1}\delta(y)dy
\end{align*}
Optimizing over all integrable $\delta(y)$ satisfying $D(f,g) \leq \Delta$ completes the proof.

\section{Proof of Corollary \ref{cor:n}}  \label{app:corn}
    Similar to the proof of Corollary~\ref{cor:q} we can lower bound the optimization problem with a convex one to obtain
    \begin{equation*} \label{eq:conv}\begin{split}
        H_{SI}^{n}(\Delta) \geq \inf_{\delta(y)} \;&\int_{0}^{1} \Breve{H}^{U}\left(\frac{\delta(y)}{\alpha}\right)dy \\
        \text{s.t } \; &\int_{0}^{1}\delta(y)dy \leq \Delta \\
        &\delta(y) \geq 0 \text{ for all } y \in [0,1].
    \end{split}
    \end{equation*}
    By symmetry, we claim that a constant function $\delta(y)=\Delta$ minimizes the above convex optimization problem. Then for any integrable $\delta'(y)$ satisfying the constraints we can write
    \begin{align*}
        \int_{0}^{1} \Breve{H}^{U}\left(\frac{\Delta}{\alpha}\right)&dy - \int_{0}^{1} \Breve{H}^{U}\left(\frac{\delta'(y)}{\alpha}\right)dy  \\
        &= \int_{0}^{1} \left(\Breve{H}^{U}\left(\frac{\Delta}{\alpha}\right) - \Breve{H}^{U}\left(\frac{\delta'(y)}{\alpha}\right)\right)dy \\
        &\leq \int_{0}^{1}\frac{\partial\Breve{H}^{U}(\Delta/\alpha)}{\partial(\Delta/\alpha)}\left(\frac{\Delta}{\alpha}-\frac{\delta'(y)}{\alpha}\right)dy = 0.
    \end{align*}
    where last line follows from convexity of $\Breve{H}^{U}(\cdot)$ and the constraints. 
\end{document}